\documentclass[11pt, letterpaper]{article}

\usepackage{amsmath,amssymb}
\usepackage{graphicx}
\usepackage{caption}
\usepackage{bm}
\usepackage[margin=1in]{geometry}
\usepackage{setspace}
\usepackage{algorithm}
\usepackage{amsthm}

\newtheorem{thm}{\bfseries Theorem}[section]
\newtheorem{lem}[thm]{\bfseries Lemma}
\newtheorem{prop}[thm]{\bfseries Proposition}

\theoremstyle{definition}
\newtheorem{defi}[thm]{\bfseries Definition}
\newtheorem{ex}[thm]{\bfseries Example}
\newtheorem{rem}[thm]{\bfseries Remark}
\newtheorem{prob}[thm]{\bfseries Problem}

\newcommand{\bsl}{\backslash}
\newcommand{\Ra}{\Rightarrow}

\newcommand{\real}{\mathbb R}

\newcommand{\set}[2]{\{ #1 \, | \, #2 \}}

\newcommand{\tcei}[1]{\lceil #1 \rceil}

\newcommand{\1}{\bm 1}
\newcommand{\gap}{\mathrm{gap}}

\makeatletter
	
	\@addtoreset{equation}{section}
\makeatother

\makeatletter
\def\Hline{%
\noalign{\ifnum0=`}\fi\hrule \@height 1.2pt \futurelet
\reserved@a\@xhline}
\makeatother

\makeatletter
\def\thanks#1{%
   \footnotemark
   \edef\@tempa{\noexpand\noexpand\noexpand\footnotetext[\the\c@footnote]}%
   \toks@\expandafter{\@thanks}%
   \toks\tw@{{#1}}
   \xdef\@thanks{\the\toks@\@tempa\the\toks\tw@}}
\makeatother

\begin{document}
\title{A polynomial time algorithm to compute geodesics \\ in CAT(0) cubical complexes}
\author{Koyo HAYASHI \\
Department of Mathematical Informatics, \\
Graduate School of Information Science and Technology, \\
University of Tokyo, Tokyo 113-8656, Japan. \\
Email: \texttt{koyo\_hayashi@mist.i.u-tokyo.ac.jp}
}
\date{}
\maketitle

\begin{quote}
\textbf{Abstract.}
This paper presents the first polynomial time algorithm to compute geodesics in a CAT(0) cubical complex in general dimension.
The algorithm is a simple iterative method to update breakpoints of a path joining two points using Miller, Owen and Provan's algorithm (2015) as a subroutine.
Our algorithm is applicable to any CAT(0) space in which geodesics between two close points can be computed, not limited to CAT(0) cubical complexes.
\end{quote}

\section{Introduction}
Computing a shortest path in a polyhedral domain in Euclidean space is a fundamental and important algorithmic problem, which is intensively studied in computational geometry~\cite{Mit}.
This problem is relatively easy to solve in the two-dimensional case; it can generally be reduced to a discrete graph searching problem where some combinatorial approaches can be applied.
In three or more dimensions, however, the problem becomes much harder; it is not even discrete.
In fact, it was proved by Canny and Reif~\cite{Can} that the shortest path problem in a polyhedral domain is NP-hard.
Mitchell and Sharir~\cite{MitSha} have shown that the problem of finding a shortest obstacle-avoiding path is NP-hard even for the case of a region with obstacles that are disjoint axis-aligned boxes.
On the other hand, there are some cases where one can obtain polynomial time complexity.
For instance, it was shown by Sharir~\cite{Sha} that a shortest obstacle-avoiding path among $k$ disjoint convex polyhedra having altogether $n$ vertices, can be found in $n^{O(k)}$ time,
which implies that this problem is polynomially solvable if $k$ is a small constant.

What determines the tractability of the shortest path problem in geometric domains?
One of promising answers to this challenging question is \emph{global non-positive curvature}, or \emph{CAT(0) property}~\cite{Gro}.
CAT(0) spaces are metric spaces in which geodesic triangles are ``not thicker'' than those in the Euclidean plane, and enjoy various fascinating properties generalizing those in Euclidean and hyperbolic spaces.
As Ghrist and LaValle~\cite{GhrLaV} observed, no NP-hard example in~\cite{MitSha} is a CAT(0) space.
One of the significant properties of CAT(0) spaces is the uniqueness of geodesics: Every pair of points can be joined by a unique geodesic.
Computational and algorithmic theory on CAT(0) spaces is itself a challenging research field~\cite{Bac}.

One of fundamental and familiar CAT(0) spaces is a \emph{CAT(0) cubical complex}.
A cubical complex is a polyhedral complex where each cell is isometric to a unit cube of some dimension and the intersection of any two cells is empty or a single face.
Gromov~\cite{Gro} gave a purely combinatorial characterization of cubical complexes of non-positive curvature as cubical complexes in which the link of each vertex is a flag simplicial complex.
Chepoi~\cite{Che} and Roller~\cite{Rol} established that the $1$-skeletons of CAT(0) cubical complexes are exactly \emph{median graphs}, i.e., graphs in which any three vertices admit a unique median vertex.
It is also shown by Barth\'{e}lemy and Constantin~\cite{Bar} that median graphs are exactly the domains of event structures~\cite{Nie}.
These nice combinatorial characterizations are one of the main reasons why CAT(0) cubical complexes frequently appear in mathematics, for instance, in geometric group theory~\cite{Rol, Sag},
metric graph theory~\cite{BanChe}, concurrency theory in computer science~\cite{Nie}, theory of reconfigurable systems~\cite{Abr, Ghr}, and phylogenetics~\cite{Bil}.

There has been several polynomial time algorithms to find shortest paths in some CAT(0) cubical complexes.
A noteworthy example is for a \emph{tree space}, introduced by Billera, Holmes and Vogtmann~\cite{Bil} as a continuous space of phylogenetic trees.
This space is shown to be CAT(0), and consequently provides a powerful tool for comparing two phylogenetic trees through the unique geodesic.
Owen and Provan~\cite{Owe, OwePro} gave a polynomial time algorithm for finding geodesics in tree spaces, which was generalized by Miller et al.~\cite{Mil} to \emph{CAT(0) orthant spaces}, i.e., complexes of Euclidean orthants that are CAT(0).
Chepoi and Maftuleac~\cite{CheMaf} gave an efficient polynomial time algorithm to compute geodesics in a two dimensional CAT(0) cubical complex.
These meaningful polynomiality results naturally lead to a question: What about arbitrary CAT(0) cubical complexes?

Ardila, Owen and Sullivant~\cite{Ard} gave a combinatorial description of CAT(0) cubical complexes, employing a poset endowed with an additional relation, called a \emph{poset with inconsistent pairs (PIP)}.
This can be viewed as a generalization of Birkhoff's theorem that gives a compact representation of distributive lattices by posets.
In fact, they showed that there is a bijection between CAT(0) cubical complexes and PIPs.
(Through the above-mentioned equivalence, this can be viewed as a rediscovery of the result of Barth\'{e}lemy and Constantin~\cite{Bar}, who found a bijection between PIPs and pointed median graphs.)
This relationship enables us to express an input CAT(0) cubical complex as a PIP:
For a poset with inconsistent pairs $P$, the corresponding CAT(0) cubical complex $\mathcal K_P$ is realized as a subcomplex of the $|P|$-dimensional cube $[0, 1]^{P}$ in which
the cells of $\mathcal K_P$ are specified by structures of $P$.
Adopting this embedding as an input, they gave the first algorithm to compute geodesics in an arbitrary CAT(0) cubical complex.
Their algorithm is based on an iterative method to update a sequence of cubes that may contain the geodesic,
where at each iteration it solves a touring problem using second order cone programming~\cite{Pol}.
They also showed that the touring problem for general CAT(0) cubical complexes has intrinsic algebraic complexity,
and geodesics can have breakpoints whose coordinates have nonsolvable Galois group.
This implies that there is no exact simple formula for the geodesic and therefore in general, one can only obtain an approximate one.
Unfortunately, even if the touring problem could be solved exactly, it is not known whether or not their algorithm is a polynomial one;
that is, no polynomial time algorithm has been known for the shortest path problem in a CAT(0) cubical complex in general dimension.

\paragraph{Main result.}
In this paper, we present the first polynomial time algorithm to compute geodesics in a CAT(0) cubical complex in general dimension,
answering the open question suggested by these previous work; namely we show that:
\begin{quote}
Given a CAT(0) cubical complex $\mathcal K$ represented by a poset with inconsistent pairs $P$ and two points $p, q$ in $\mathcal K$,
one can find a path joining $p$ and $q$ of length at most $d(p, q) + \epsilon$ in time polynomial in $|P|$ and $\log(1/\epsilon)$.
\end{quote}
The algorithm is quite simple, without depending on any involved techniques such as semidefinite programming.
To put it briefly, our algorithm first gives a polygonal path joining $p$ and $q$ with a fixed number ($n$, say) of breakpoints, and then iteratively updates the breakpoints of the path
until it becomes a desired one.
To update them, we compute the midpoints of the two close breakpoints by using Miller, Owen and Provan's algorithm. The resulting number of iterations is bounded by a polynomial in $n$.
Key tools that lead to this bound are linear algebraic techniques and the convexity of the metric of CAT(0) spaces, rather than inherent properties of cubical complexes.
Due to its simplicity, our algorithm is applicable to any CAT(0) space where geodesics between two close points can be found, not limited to CAT(0) cubical complexes.
We believe that our result will be an important step toward developing computational geometry in CAT(0) spaces.

\paragraph{Application.}
A \emph{reconfigurable system}~\cite{Abr, Ghr} is a collection of states which change according to local and reversible moves that affect global positions of the system.
Examples include robot motion planning, non-collision particles moving around a graph, and protein folding; see~\cite{Ghr}.
Abrams, Ghrist and Peterson~\cite{Abr, Ghr} considered a continuous space of all possible positions of a reconfigurable system, called a \emph{state complex}. 
Any state complex is a cubical complex of non-positively curved~\cite{Ghr}, and it becomes CAT(0) in many situations.
In the robotics literature, geodesics (in the $l_2$-metric) in the CAT(0) state complex corresponds to the motion planning to get the robot from one position to another one with minimal power consumption.
Our algorithm enables us to find such an optimal movement of the robot in polynomial time.

\paragraph{Organization.}
The rest of this paper is organized as follows.
In Section \ref{sec:CAT0geodesics}, we devise an algorithm to compute geodesics in general CAT(0) spaces.
In Section \ref{sec:geoCube}, we present a polynomial time algorithm to compute geodesics in CAT(0) cubical complexes, using the result of Section \ref{sec:CAT0geodesics}.
Section \ref{subsec:CAT0spaces} and Section \ref{subsec:cube} to \ref{subsec:orthant} are preliminary sections, where CAT(0) metric spaces, CAT(0) cubical complexes, median graphs, PIPs and CAT(0) orthant spaces are introduced.


\section{Computing geodesics in CAT(0) spaces}\label{sec:CAT0geodesics}
In this section we devise an algorithm to compute geodesics in general CAT(0) spaces, not limited to CAT(0) cubical complexes.
\begin{figure}[t]
	\centering
	\includegraphics[bb= 0 0 312 102]{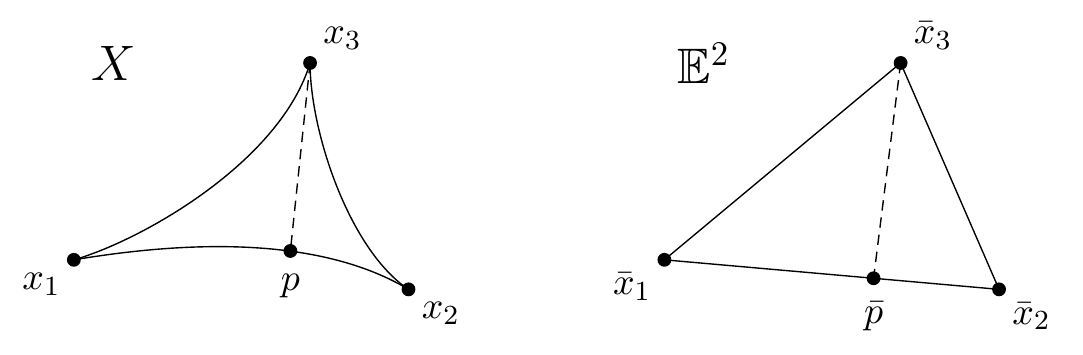}\caption{CAT(0) space.}
	\label{fig:thin}
\end{figure}

\subsection{CAT(0) space}\label{subsec:CAT0spaces}
Let $(X, d)$ be a metric space.
A \textit{geodesic} joining two points $x , y \in X$ is a map $\gamma :[0,1] \to X$
such that $\gamma(0) = x, \gamma(1) = y$ and $d(\gamma(s), \gamma(t)) = d(x,y)|s - t|$ for all $s, t \in [0, 1]$.
The image of $\gamma$ is called a \textit{geodesic segment} joining $x$ and $y$.
A metric space $X$ is called \textit{(uniquely) geodesic} if every pair of points $x, y \in X$ is joined by a (unique) geodesic.

For any triple of points $x_1, x_2, x_3$ in a metric space $(X, d)$, there exists a triple of points $\bar x_1, \bar x_2, \bar x_3$ in the Euclidean plane $\mathbb E^2$ such that
$d(x_i, x_j) = d_{\mathbb E^2}(\bar x_i, \bar x_j)$ for $i, j \in \{1, 2, 3\}$.
The Euclidean triangle whose vertices are $\bar x_1, \bar x_2$ and $\bar x_3$ is called a \textit{comparison triangle} for $x_1, x_2, x_3$.
(Note that such a triangle is unique up to isometry.)
A geodesic metric space $(X, d)$ is called a \textit{CAT(0) space} if for any $x_1, x_2, x_3 \in X$ and any $p$ belonging to a geodesic segment joining $x_1$ and $x_2$,
the inequality $d(x_3, p) \le d_{\mathbb E^2}(\bar x_3, \bar p)$ holds, where $\bar p$ is the unique point in $\mathbb E^2$ satisfying $d(\bar x_i, \bar p) = d_{\mathbb E^2}(x_i, p)$ for $i = 1, 2$. See Figure~\ref{fig:thin}.

This simple definition yields various significant properties of CAT(0) spaces; see~\cite{Bri} for details.
One of the most basic properties of  CAT(0) spaces is the convexity of the metric.
A geodesic metric space $(X, d)$ is said to be \textit{Busemann convex} if for any two geodesics $\alpha, \beta : [0, 1] \to X$,
the function $f : [0, 1] \to \real$ given by $f(t) := d(\alpha(t), \beta(t))$ is convex.
\begin{lem}[{\cite[Proposition II.2.2]{Bri}}]\label{lem:Busemann}
Every CAT(0) space is Busemann convex.
\end{lem}
A Busemann convex space $X$ is uniquely geodesic.
Indeed, for any two geodesics $\alpha, \beta : [0, 1] \to X$ with $\alpha(0) = \beta(0)$ and $\alpha(1) = \beta(1)$, one can easily see that $\alpha$ and $\beta$ coincide, since
$d(\alpha(t), \beta(t)) \le (1-t) d(\alpha(0), \beta(0)) + t d(\alpha(1), \beta(1)) = 0$ for all $t \in [0, 1]$.
This implies that:
\begin{thm}[{\cite[Proposition II.1.4]{Bri}}]
Every CAT(0) space is uniquely geodesic.
\end{thm}

\subsection{Algorithm}\label{subsec:CAT0algo}
\begin{figure}[t]
	\centering
	\includegraphics[bb=0 0 363 184]{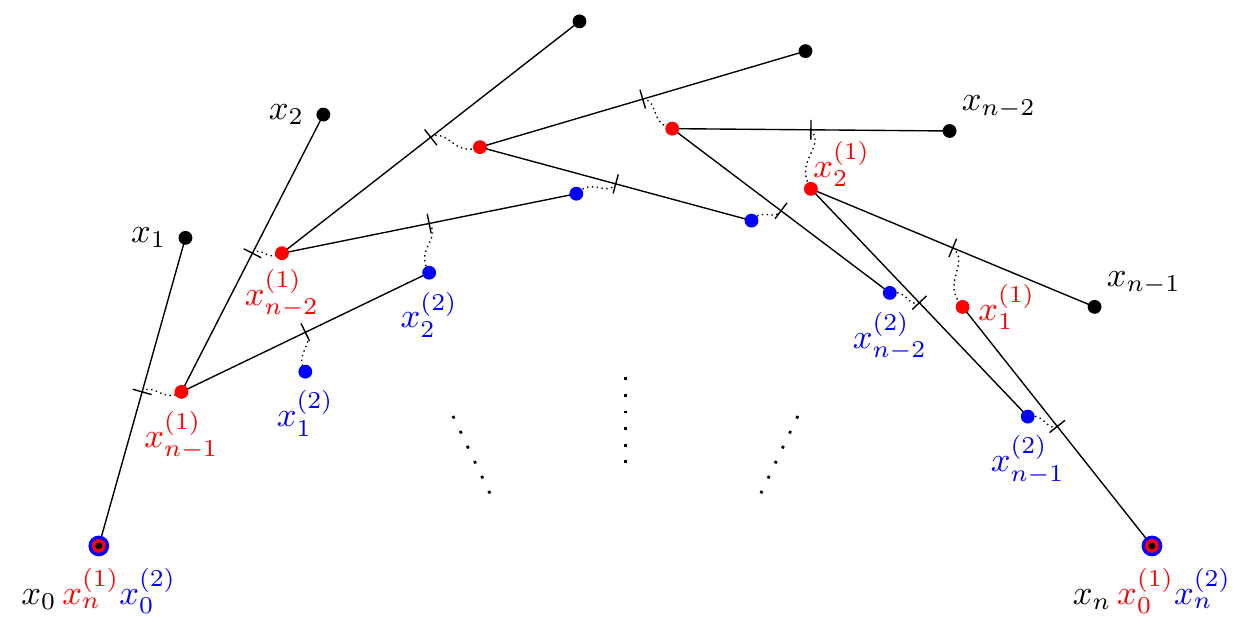}
	\caption{An illustration of Algorithm 1.}
	\label{fig:FigAlgo}
\end{figure}

Let $X$ be a CAT(0) space.
We shall refer to an element $x$ in the product space $X^{n+1}$ as a \textit{chain}, and write $x_{i-1}$ to denote the $i$-th component of $x$, i.e., $x = (x_0, x_1, \dots, x_n)$.
For any chain $x \in X^{n+1}$, we define the \textit{length} of $x$ by $\sum_{i=0}^{n-1}d(x_i, x_{i+1})$ and denote it by $\ell (x)$.
We consider the following problem:
\begin{equation}\label{eq:prob1}
\text{\parbox{.85\textwidth}{
Given two points $p, q \in X$, a chain $x \in X^{n+1}$ with $x_0 = p$ and $x_n = q$, and a positive parameter $\epsilon > 0$, find a chain $y \in X^{n+1}$ such that $y_0 = p$, $y_n=q$ and $\ell (y) \le d(p, q) + \epsilon$,
}}
\end{equation}
under the situation where we are given an oracle to perform the following operation for some $D > 0$:
\begin{equation}\label{eq:oracle}
\text{\parbox{.85\textwidth}{
Given two points $p, q \in X$ with $d(p, q) \le D$, compute the geodesic joining $p$ and $q$ in arbitrary precision.
}}
\end{equation}

To explain our algorithm to solve this problem, we need some definitions.
Since $X$ is uniquely geodesic, every pair of points $p, q \in X$ has a unique midpoint $w$ satisfying $2d(w, p) = 2d(q, w) = d(p, q)$.
For a nonnegative real number $\delta \ge 0$, a $\delta$-\textit{midpoint} of $p$ and $q$ is a point $w' \in X$ satisfying $d(w', w) \le \delta$, where $w$ is the midpoint of $p$ and $q$.
\begin{defi}[$\delta$-halved chain]
Let $\delta$ be a nonnegative real number.
For any chain $x \in X^{n+1}$, a chain $z \in X^{n+1}$ is called a $\delta$-\textit{halved chain} of $x$ if it satisfies the following:
\begin{quote}
$z_0 = x_n$, $z_n = x_0$ and $z_i$ is a $\delta$-midpoint of $z_{i+1}$ and $x_{n-i}$ for $i = 1,2, \dots, n-1$.
\end{quote}
For an integer $k \ge 0$, we say that $x^{(k)}$ is a \textit{$k$-th $\delta$-halved chain} of $x$ if there exists a sequence $\{x^{(j)}\}_{j=0}^{k}$ of chains in $X^{n+1}$ such that $x^{(0)} = x$ and $x^{(j)}$ is a $\delta$-halved chain of $x^{(j-1)}$ for $j = 1, 2, \dots, k$.
\end{defi}

Our algorithm can be described as follows. To put it briefly, the algorithm just finds a $k$-th $\delta$-halved chain of a given chain $x$ for some large $k$ and small $\delta$;
see Figure~\ref{fig:FigAlgo} for an illustration.
In the algorithm the local optimization is done alternatively ``from left to right'' and ``from right to left'' so that the analysis will be easier.

\begin{algorithm}
\caption{}
\label{algo:1}
\textbf{Input.} Two points $p, q \in X$, a chain $x \in X^{n+1}$ with $x_0 = p$ and $x_n = q$, and parameters $\epsilon > 0, \delta \ge 0$.
\begin{enumerate}
\renewcommand{\labelenumi}{\theenumi}
\renewcommand{\theenumi}{$\langle$\arabic{enumi}$\rangle$}
\renewcommand{\labelenumii}{\theenumii}
\renewcommand{\theenumii}{$\langle$2-\arabic{enumii}$\rangle$}
\setlength{\parskip}{0cm}
\setlength{\itemsep}{0cm}
	\item Set $x^{(0)} := x$. 
	\item\label{item:algo2} For $j = 0, 1, 2, \dots$, do the following:
	\begin{enumerate}
\setlength{\parskip}{0cm}
\setlength{\itemsep}{0cm}
		\item Set $z_0 := x_n^{(j)}$ and $z_n := x_0^{(j)}$.
		\item For $i = 1, 2, \dots, n-1$, do the following:
	\begin{equation}\label{eq:algomid}
		\text{Compute a $\delta$-midpoint $w$ of $z_{n-i+1}$ and $x_i^{(j)}$, and set $z_{n-i} := w$}.
	\end{equation}
		\item Set $x^{(j+1)} := (z_0, z_1, \dots, z_n)$.
	\end{enumerate}
\end{enumerate}
\end{algorithm}

For any chain $x \in X^{n+1}$, define the \textit{gap} of $x$ by $\max \{ d(x_0, x_1), \max_{1 \le i \le n-1}2 d(x_i, x_{i+1}) \}$ and denote it by $\gap (x)$.
The following theorem states that Algorithm~1 solves problem~(\ref{eq:prob1}).
\begin{thm}\label{thm:main1}
Let $p, q \in X$ be given two points, $x \in X^{n+1}$ be a given chain with $x_0 = p$ and $x_n = q$,
and $\epsilon > 0, 0\le \delta \le \epsilon/(16n^3)$ be parameters.
\begin{enumerate}
\renewcommand{\labelenumi}{\theenumi}
\renewcommand{\theenumi}{\textup{(\roman{enumi})}}
\setlength{\parskip}{0cm}
\setlength{\itemsep}{0cm}
\item\label{item:1} For $j \ge n^2 \log(4n \cdot \ell (x)/\epsilon)$, one has $\ell (x^{(j)}) \le d(p, q) + \epsilon$.
\item\label{item:2} If $\gap(x) \le D/2 - \epsilon$ for some $D > 0$, then for all $j \ge 0$ and for $i = 1, 2, \dots, n-1$, one has $d(z_{n-i+1}, x_i^{(j)}) \le D$ in (\ref{eq:algomid}).
\end{enumerate}
In particular, for $\gap(x) \le D/2 - \epsilon$, one can find a chain $y \in X^{n+1}$ such that $y_0 = p$, $y_n=q$ and $\ell (y) \le d(p, q) + \epsilon$, with $O(n^3 \log (nD/\epsilon))$ calls of an oracle to perform (\ref{eq:oracle}).
\end{thm}

\begin{ex}
We give an example of CAT(0) spaces to which our algorithm is applicable.
A \textit{$B_2$-complex} is a two dimensional piecewise Euclidean complex in which each 2-cell is isomorphic to an isosceles right triangle with short side of length one~\cite{Ger}.
A CAT(0) $B_2$-complex is called a \textit{folder complex}~\cite{Che}; see Figure~\ref{fig:FigFolder} for an example.
\begin{figure}[t]
	\centering
	\includegraphics[bb=0 0 281 96]{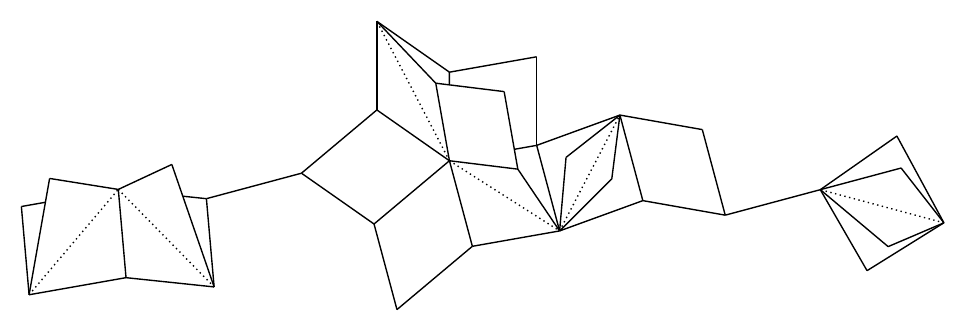}
	\caption{A folder complex.}
	\label{fig:FigFolder}
\end{figure}
One can show that for a folder complex $\mathcal F$, computing the geodesic between two points $p, q \in \mathcal F$ with $d(p, q) \le 1$ can be reduced to an easy calculation on
a subcomplex of $\mathcal F$ having a few cells.
This implies that our algorithm enables us to find geodesics between two points in a folder complex $\mathcal F$ in time bounded by a polynomial in the size of $\mathcal F$.
\end{ex}

\subsection{Analysis}\label{subsec:analysis}
For any chain $x \in X^{n+1}$,
we define the \textit{reference chain} $\hat x \in X^{n+1}$ of $x$ as follows: $\hat x_0 := x_0$ and $\hat x_i := \gamma((i+1)/(n+1))$ for $i = 1,2, \dots, n$,
where $\gamma : [0, 1] \to X$ is the geodesic with $\gamma(0) = x_0$ and $\gamma(1) = x_n$.
Reference chains are designed not to be equally spaced but to have a double gap in the beginning so that the analysis of the algorithm will be easier.
Note that the reference chain $\hat x$ of $x$ is determined just by its end components $x_0, x_n$, and therefore
for any chain $x$ and any even $\delta$-halved chain $x^{(2k)}$ of $x$ their reference chains coincide: $\hat x^{(2k)} = \hat x$.
A key observation that leads to Theorem~\ref{thm:main1} is that: For any chain $x \in X^{n+1}$ and any $k$-th $\delta$-halved chain $x^{(k)}$ of $x$ with $k$ sufficiently large and $\delta$ sufficiently small,
the distance between $x^{(k)}$ and its reference chain $\hat x^{(k)}$ is small enough for its length $\ell (x^{(k)})$ to approximate well $d(x_0, x_n)$;
moreover, the value of such a $k$ can be bounded by a polynomial in $n$.
The next lemma states this fact.

\begin{lem}\label{lem:convergence}
Let $x \in X^{n+1}$.
Any $k$-th $\delta$-halved chain $x^{(k)}$ of $x$ satisfies
\begin{equation*}
	d(x_i^{(k)}, \hat x^{(k)}_i) \le (5/4) \ell (x) e^{-k/n^2}  + 3 n^2\delta
\end{equation*}
for $i = 1, 2, \dots, n-1$, where $e$ is the base of the natural logarithm.
\end{lem}
\begin{proof}
Let $\{ x^{(j)} \}_{j \ge 0}$ be a sequence of chains in $X^{n+1}$ such that $x^{(0)} = x$ and $x^{(j)}$ is a $\delta$-halved chain of $x^{(j-1)}$ for $j \ge 1$.
Fix an integer $1 \le i \le n-1$ and an integer $k \ge 0$.
Note that by definition $x_i^{(k+1)}$ is a $\delta$-midpoint of $x_{i+1}^{(k+1)}$ and $x_{n-i}^{(k)}$ and that $\hat x^{(k+1)}_i$ is the midpoint of $\hat x^{(k+1)}_{i+1}$ and $\hat x^{(k)}_{n-i}$.
Hence, by Lemma~\ref{lem:Busemann} and the triangle inequality, we have
\begin{align}\label{eq:mid_inequality}\begin{split}
	2d(x^{(k+1)}_i, \hat x^{(k+1)}_i ) &\le 2d(w, \hat x^{(k+1)}_i ) + 2\delta \\
&\le d( x^{(k+1)}_{i+1}, \hat x^{(k+1)}_{i+1}) + d(x^{(k)}_{n-i}, \hat x^{(k)}_{n-i}) + 2\delta,
\end{split}\end{align}
where $w$ is the midpoint of $x^{(k+1)}_{i+1}$ and $x^{(k)}_{n-i}$. See Figure~\ref{fig:FigProof1} for intuition.
\begin{figure}[t]
	\centering
	\includegraphics[bb= 0 0 295 111]{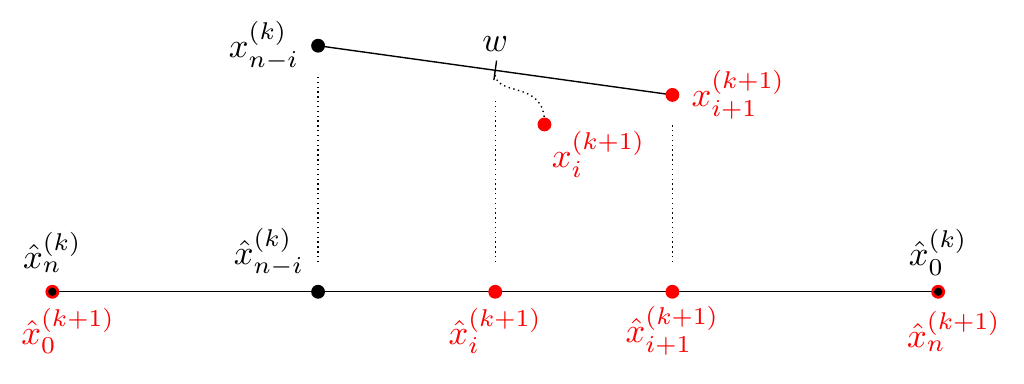}
	\caption{To the proof of Lemma~\ref{lem:convergence}. The chain $x^{(j)}$ is a $j$-th $\delta$-halved chain of $x$ for $j = k, k+1$.}
	\label{fig:FigProof1}
\end{figure}

Let $v^{(k)}$ be a column vector of dimension $n-1$ whose $i$-th entry equals $d(x^{(k)}_i, \hat x^{(k)}_i)$ for $i = 1, 2, \dots, n-1$.
Let $J$ be a square matrix of order $n-1$ whose $(i, j)$ entry equals $1$ if $i + j = n$ and $0$ otherwise.
Let $K$ be a square matrix of order $n-1$ whose $(i, j)$ entry equals $1$ if $j = i + 1$ and $0$ otherwise.
Then, by (\ref{eq:mid_inequality}) we have $2v^{(k+1)} \le K v^{(k+1)} + J v^{(k)} + 2 \delta \1$ for each $k \ge 0$,
where $\1$ is a column vector with all entries equal to $1$.
Let $A_{n-1}$ be a square matrix of order $n-1$ whose $(i, j)$ entry equals $(1/2)^{n+1-i-j}$ if $i + j \le n$ and $0$ otherwise.
Then one can easily see that $(2I - K)^{-1}J = A_{n-1}$.
Hence we have
\begin{equation}\label{eq:recursion}
	v^{(k+1)} \le A_{n-1} v^{(k)} + A_{n-1} J^{-1} (2 \delta \1) \le A_{n-1} v^{(k)} + 2 \delta \1
\end{equation}
for each $k \ge 0$.
We show that
\begin{equation}\label{eq:bounds}
	v^{(k)} \le ((5/4)\ell (x)e^{-k/n^2} + 3n^2 \delta )\1
\end{equation}
for any integer $k \ge 0$.
The inequality (\ref{eq:recursion}) inductively yields that $v^{(k)} \le (A_{n-1})^k v^{(0)} + 2 \delta(I + A_{n-1} + \dots + (A_{n-1})^{k-1}) \1 \le \ell (x) (A_{n-1})^k \1 + 2 \delta (I-A_{n-1})^{-1} \1$.
Here, the inequality $v^{(0)} \le \ell (x) \1$ comes from the triangle inequality. Indeed, we have
\begin{align*}
d(x_i, \hat x_i) &\le \min \{ d(x_0, \hat x_i) + \textstyle\sum_{j=0}^{i-1} d(x_j, x_{j+1}), \ d(\hat x_i, x_n) + \textstyle\sum_{j=i}^{n-1} d(x_j, x_{j+1}) \} \\
&\le (d(x_0, x_n) + \ell (x))/2 \le \ell (x)
\end{align*}
for $i = 1, 2, \dots, n-1$.
In Lemma~\ref{lem:sp} below, we prove $(I-A_{n-1})^{-1} \1 \le (5(n-1)^2/4) \1$ (for $n- 1 \ge 2$).
This yields that $(I-A_{n-1})^{-1} \1 \le (3/2)n^2 \1$ for $n \ge 2$.
Also, we prove $(A_{n-1})^{k} \1 \le (5/4)e^{-k/(n-1)^2} \1$ (for $n-1 \ge 2$) in Lemma~\ref{lem:sp}.
This implies that $(A_{n-1})^k \1 \le (5/4) e^{-k/n^2} \1$ for $n \ge 2$.
This proves (\ref{eq:bounds}) and therefore completes the proof of the lemma.
\end{proof}

Let us now prove Theorem~\ref{thm:main1}.

\begin{proof}[Proof of Theorem~\ref{thm:main1}]
We may assume that $n \ge 2$.
We first show~\ref{item:1}. If $\delta \le \epsilon/(16n^3)$ and $j \ge n^2 \log(4n \cdot \ell (x)/\epsilon)$, then by Lemma~\ref{lem:convergence}, any $j$-th $\delta$-halved chain $x^{(j)}$ of $x$ satisfies $d(x^{(j)}_i, \hat x^{(j)}_i) \le 5\epsilon/(16n) + 3\epsilon/(16n) = \epsilon/(2n)$
for $i = 1, 2, \dots, n-1$. Hence one has
\begin{align}\label{eq:reference_approximation}\begin{split}
d(x_i^{(j)}, x_{i+1}^{(j)}) &\le d(x_i^{(j)}, \hat x_i^{(j)}) + d(\hat x_i^{(j)}, \hat x_{i+1}^{(j)}) + d(\hat x_{i+1}^{(j)}, x_{i+1}^{(j)}) \\
&\le d(\hat x_i^{(j)}, \hat x_{i+1}^{(j)}) + \epsilon/n
\end{split}\end{align}
for $i = 0, 1, \dots, n-1$.
This implies that
$\ell (x^{(j)}) = \sum_{i = 0}^{n-1} d(x_i^{(j)}, x_{i+1}^{(j)}) \le \sum_{i = 0}^{n-1}( d(\hat x_i^{(j)}, \hat x_{i+1}^{(j)}) + \epsilon/n ) = d(x_0, x_n) + \epsilon = d(p, q) + \epsilon$,
and therefore completes the proof of~\ref{item:1}.

To prove~\ref{item:2}, we first show
\begin{equation}\label{eq:gap_bound}
	d(z_{n-i+1}, x_i^{(j)}) \le \gap(x^{(j)}) + 2 \delta \quad (i = 1, 2, \dots, n; j \ge 0),
\end{equation}
by induction on $i$.
The case $i = 1$ being trivial, suppose that $i \ge 2$.
Since $z_{n-i+1}$ is a $\delta$-midpoint of $z_{n-i+2}$ and $x_{i-1}^{(j)}$, the triangle inequality and the induction yield
$d(z_{n-i+1}, x_{i}^{(j)}) \le \delta + d(z_{n-i+2}, x_{i-1}^{(j)})/2 + d(x_{i-1}^{(j)}, x_{i}^{(j)}) \le \delta + (\gap(x^{(j)})/2 + \delta) + \gap(x^{(j)})/2 = \gap(x^{(j)}) + 2 \delta$,
which completes the induction.
See Figure~\ref{fig:induction} for intuition.
\begin{figure}[t]
	\centering
	\includegraphics[bb=0 0 323 162]{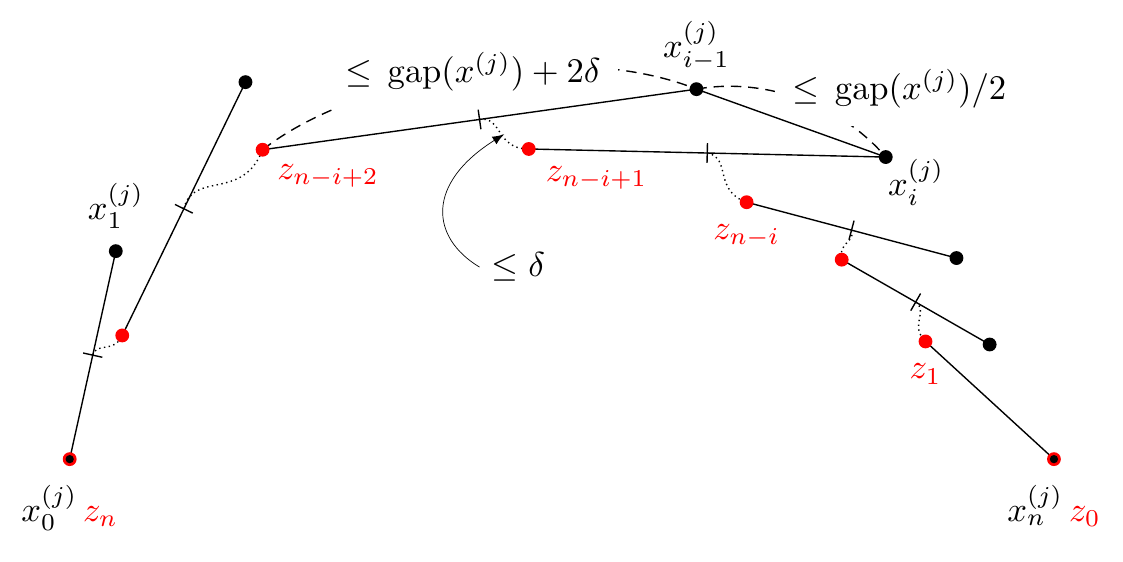}
	\caption{To the proof of Theorem~\ref{thm:main1}~\ref{item:2}. The induction hypothesis $d(z_{n-i+2}, x_{i-1}^{(j)}) \le \gap(x^{(j)}) +2 \delta$ and the triangle inequality yield the next step $d(z_{n-i+1}, x_{i}^{(j)}) \le \gap(x^{(j)}) +2 \delta$.}
	\label{fig:induction}
\end{figure}

It follows from (\ref{eq:gap_bound}) that $\gap(x^{(j+1)}) \le \gap(x^{(j)}) + 4 \delta$ for $j \ge 0$.
Indeed, the case $i = n$ in (\ref{eq:gap_bound}) implies that $d(z_1, z_0) = d(z_1, x^{(j)}_n) \le \gap(x^{(j)}) + 2 \delta$; on the other hand,
by the triangle inequality and (\ref{eq:gap_bound}), one has $d(z_{n-i+1}, z_{n-i}) \le d(z_{n-i+1}, x^{(j)}_{i})/2 + \delta \le \gap(x^{(j)})/2 + 2 \delta$
for $i = 1, 2, \dots, n-1$.
Thus, one has $\gap(x^{(j+1)}) \le \max \{ \gap(x^{(j)}) + 2 \delta, 2 (\gap(x^{(j)})/2 + 2 \delta) \} = \gap(x^{(j)}) + 4 \delta$.

The inequality (\ref{eq:gap_bound}) implies that in order to prove~\ref{item:2} it suffices to show that $\gap(x^{(j)}) + 2 \delta \le D$ for all $j \ge 0$.
Suppose that $\delta \le \epsilon / (16n^3)$. We consider two cases.

\noindent\textbf{Case 1}: $j \le n^2 \log(4n \cdot \ell (x)/\epsilon)$.
Note that $\ell (x) \le n \cdot \gap(x)$ and that $\gap(x^{(j)}) \le \gap(x) + 4j\delta$.
However roughly one estimates an upper bound of $4j\delta$, one can get
\begin{equation*}
4j\delta \le 4 \cdot \frac{\epsilon}{16 n^3} \cdot n^2 \log \frac{4n^2 \cdot \gap(x)}{\epsilon} = \frac{\epsilon}{4 n} \left( \log \frac{\gap(x)}{\epsilon} +2 \log 2n \right ) \le \frac{\gap(x)}{4 n e} + \frac{\epsilon}{e},
\end{equation*}
where the last inequality comes from the fact that $\log t \le t/e$ for any $t > 0$.
It is easy to see that $\gap(x^{(j)}) + 2 \delta \le \gap(x) + \gap(x)/(4ne) + \epsilon/e + \epsilon/(8n^3) \le D$, provided $\gap(x) \le D/2 - \epsilon$.

\noindent\textbf{Case 2}: $j \ge n^2 \log(4n \cdot \ell (x)/\epsilon)$.
Recall (\ref{eq:reference_approximation}).
Since $d(x_0, x_n)/(n+1) \le \gap(x)/2$, we have
\begin{equation*}
\gap(x^{(j)}) \le \max \{ \gap(x) + \epsilon/n, 2(\gap(x)/2 + \epsilon/n) \} = \gap(x) + 2\epsilon/n.
\end{equation*}
It is easy to see that $\gap(x^{(j)}) + 2 \delta \le \gap(x) + 2\epsilon/n + \epsilon/(8n^3) \le D$, provided $\gap(x) \le D/2 - \epsilon$.

From \ref{item:1} and \ref{item:2},
we can show that if $\gap(x) \le D/2 - \epsilon$, then one can find a chain $y \in X^{n+1}$ satisfying $y_0 = p, y_n = q$ and $\ell (y) \le d(p, q) + \epsilon$, with $O(n^3 \log (nD/\epsilon))$ oracle calls.
Indeed, for $k := \tcei{n^2 \log(4n \cdot \ell (x)/\epsilon)}$, one can find a $k$-th $\delta$-halved chain $x^{(k)}$ of $x$ with $O(nk) = O(n^3 \log (nD/\epsilon))$ oracle calls, from~\ref{item:2};
its length $\ell (x^{(k)})$ is at most $d(p, q) + \epsilon$, from~\ref{item:1}.
\end{proof}

We end this section by showing the lemma used in the proof of Lemma~\ref{lem:convergence}.
Let $A_n$ be an $n \times n$ matrix whose $(i, j)$ entry is defined by
\begin{equation}\label{eq:An}
	(A_n)_{ij} := \begin{cases} (1/2)^{n+2-i-j} & (i + j \le n + 1), \\ 0 & (\text{otherwise}) \end{cases}
\end{equation}
for $i , j = 1, 2, \dots, n$.
Since $A_n$ is a nonnegative matrix, its spectral radius $\rho(A_n)$ is at most the maximum row sum of $A_n$, which immediately yields that $\rho(A_n) \le 1 - (1/2)^n$.
This inequality, however, is not tight unless $n =1$.
In fact, one can obtain a more useful upper bound of $\rho(A_n)$.
\begin{lem}\label{lem:sp}
Let $n > 1$ be an integer, and let $A_n$ be an $n \times n$ matrix defined by (\ref{eq:An}).
Then its spectral radius $\rho(A_n)$ is at most $1 - 1/n^2$.
In addition, one has $(I-A_n)^{-1} \1 \le (5n^2/4) \1$ and $(A_n)^{k} \1 \le (5/4)e^{-k/n^2} \1$ for any integer $k \ge 0$.
\end{lem}
\begin{proof}
Let $A := A_n$ for simplicity.
Let $u$ be a positive column vector of dimension $n$ whose $k$-th entry is defined by $u_k := k(n-k) + n^2$ for $k = 1, 2, \dots, n$.
By the Collatz--Wielandt inequality, in order to show $\rho(A) \le 1 - 1/n^2$ it suffices to show that $A u \le (1-1/n^2) u$.
The $k$-th entry of the vector $A u$ is 
\begin{equation*}
(Au)_k = \sum_{j=1}^{n+1-k} \frac{u_j}{2^{n+2-k-j}} = \frac{1}{2^{n+2-k}} \sum_{j=1}^{n+1-k} 2^j (-j^2 + nj + n^2).
\end{equation*}
Hence, using the general formulas
\begin{equation*}
	\sum_{j=1}^{m} j \cdot 2^j = 2 + 2^{m+1}(m-1)\quad \text{and}\quad  \sum_{j=1}^{m} j^2 \cdot 2^j = -6 + 2^{m+1}((m-1)^2 + 2),
\end{equation*}
we have
\begin{equation*}
	(Au)_k = u_k - 2 - \frac{n^2 - n - 3}{2^{n+1-k}}.
\end{equation*}
It is easy to see that for $n \ge 2$ and $1 \le k \le n$ one has
\begin{equation*}
	\frac{u_k}{n^2} = 1 + \frac{k(n-k)}{n^2} \le \frac{5}{4} \le \left( 2 - \frac{1}{2^{n+1-k}} \right) + \frac{(n-2)(n+1)}{2^{n+1-k}},
\end{equation*}
which implies that
\begin{equation*}
	\frac{u_k}{n^2} \le 2 + \frac{n^2 - n - 3}{2^{n+1-k}} \quad (k = 1, 2, \dots, n).
\end{equation*}
This completes the proof of the inequality $A u \le (1 - 1/n^2)u $.

Let us show the latter part of the lemma.
Note that $\1 \le (1/n^2)u \le (5/4) \1$. Since $(1/n^2)u \le (I - A) u$ and $(I - A)^{-1}$ is a nonnegative matrix (as $\rho(A) < 1$),
we have $(I-A)^{-1} \1 \le (1/n^2)(I-A)^{-1} u \le u \le  (5n^2/4) \1$.

Since $Au \le (1 - 1/n^2)u \le e^{-1/n^2}u$, we have $A^{k} u \le e^{-k/n^2} u$ for any integer $k \ge 0$.
Hence, $A^{k} \1 \le (1/n^2)A^{k} u\le (1/n^2)e^{-k/n^2}u \le (5/4)e^{-k/n^2}\1$.
\end{proof}

\begin{rem}
In proving Theorem~\ref{thm:main1}, we utilized only the convexity of the metric of $X$.
Hence our algorithm works even when $X$ is a Busemann convex space.
\end{rem}


\section{Computing geodesics in CAT(0) cubical complexes}\label{sec:geoCube}
In this section we give an algorithm to compute geodesics in CAT(0) cubical complexes, with an aid of the result of the preceding section.
In Section~\ref{subsec:cube} to \ref{subsec:orthant}, we recall CAT(0) cubical complexes, median graphs, PIPs and CAT(0) orthant spaces.
Section~\ref{subsec:main} is devoted to proving our main theorem.

\subsection{CAT(0) cubical complex}\label{subsec:cube}

A \textit{cubical complex} $\mathcal K$ is a polyhedral complex
where each $k$-dimensional cell is isometric to the unit cube $[0, 1]^k$ and the intersection of any two cells is empty or a single face.
The \textit{underlying graph} of $\mathcal K$ is the graph $G(\mathcal K) = (V(\mathcal K), E(\mathcal K))$,
where $V(\mathcal K)$ denotes the set of \textit{vertices} ($0$-dimensional faces) of $\mathcal K$ and $E(\mathcal K)$ denotes the set of \textit{edges} ($1$-dimensional faces) of $\mathcal K$.

A cubical complex $\mathcal K$ has an intrinsic metric induced by the $l_2$-metric on each cell.
For two points $p, q \in \mathcal K$,
a \textit{string} in $\mathcal K$ from $p$ to $q$ is a sequence of points $p = x_0, x_1, \dots, x_{m-1}, x_m = q$ in $\mathcal K$ such that for each $i = 0, 1, \dots, m-1$ there exists a cell
$C_i$ containing $x_i$ and $x_{i+1}$, and its \textit{length} is defined to be $\sum_{i=0}^{m-1}d(x_i, x_{i+1})$, where $d(x_i, x_{i+1})$ is measured inside $C_i$ by the $l_2$-metric.
The distance between two points $p, q \in \mathcal K$ is defined to be the infimum of the lengths of strings from $p$ to $q$.

Gromov~\cite{Gro} gave a combinatorial criterion which allows us to easily decide whether or not a cubical complex $\mathcal K$ is non-positively curved.
The \textit{link} of a vertex $v$ of $\mathcal K$ is the abstract simplicial complex whose vertices are the edges of $\mathcal K$ containing $v$ and where $k$ edges $e_1, \dots, e_k$ span
a simplex if and only if they are contained in a common $k$-dimensional cell of $\mathcal K$.
An abstract simplicial complex $\mathcal L$ is called \textit{flag} if any set of vertices is a simplex of $\mathcal L$ whenever each pair of its vertices spans a simplex.

\begin{thm}[Gromov~\cite{Gro}]\label{thm:Gromov}
A cubical complex $\mathcal K$ is CAT(0) if and only if $\mathcal K$ is simply connected and the link of each vertex is flag.
\end{thm}

\subsection{Median graph}\label{subsec:median}
Let $G = (V, E)$ be a simple undirected graph.
The distance $d_G(u, v)$ between two vertices $u$ and $v$ is the length of a shortest path between $u$ and $v$.
The \textit{interval} $I_G(u, v)$ between $u$ and $v$ is the set of vertices $w \in V$ with $d_G(u, v) = d_G(u, w) + d_G(w, v)$.
A vertex subset $U \subseteq V$ is said to be \textit{gated} if for every vertex $v \in V$,
there exists a unique vertex $v' \in U$, called the \textit{gate} of $v$ in $U$, such that $v' \in I_G(u, v)$ for all $u \in U$.
Every gated subset is convex, where a vertex subset $U \subseteq V$ is said to be \textit{convex} if $I_G(u, v)$ is contained in $U$ for all $u, v \in U$.
A vertex subset $H \subseteq V$ is called a \textit{halfspace} of $G$ if both $H$ and its complement $V\bsl H$ are convex.
A graph $G$ is called a \textit{median graph} if for all $u, v, w \in V$ the set $I_G(u, v) \cap I_G(v, w) \cap I_G(w, u)$ contains exactly one element, called the \textit{median} of $u, v, w$.
Median graphs are connected and bipartite.
In median graphs $G$, every convex set $S$ of $G$ is gated.
(Indeed, for each $v \in V$ one can take a vertex $v' \in S$ such that $I_G(v', v) \cap S = \{v' \}$.
Then for any $u \in S$ the median $m$ of $u, v, v'$ should be $v'$, as $m \in I_G(u, v') \subseteq S$ and $m \in I_G(v', v)$.
This implies that $v'$ is the gate of $v$ in $S$.)
Thus, in median graphs gated sets and convex sets coincide.
A \textit{median complex} is a cubical complex derived from a median graph $G$ by replacing all cube-subgraphs of $G$ by solid cubes.
It has been shown independently by Chepoi~\cite{Che} and Roller~\cite{Rol} that median complexes and CAT(0) cubical complexes constitute the same objects:
\begin{thm}[{Chepoi~\cite{Che}, Roller~\cite{Rol}}]\label{thm:Chepoi}
The underlying graph of every CAT(0) cubical complex is a median graph, and conversely, every median complex is a CAT(0) cubical complex.
\end{thm}

For a cubical complex $\mathcal K$ and any $S \subseteq V(\mathcal K)$, we denote by $\mathcal K(S)$ the subcomplex of $\mathcal K$ induced by $S$.
The following property of CAT(0) cubical complexes is particularly important for us.
\begin{thm}[{\cite[Proposition 1]{CheMaf}}]\label{thm:CheMaf}
Let $\mathcal K$ be a CAT(0) cubical complex.
For any convex set $S$ of the underlying graph $G(\mathcal K)$, 
the subcomplex $\mathcal K(S)$ induced by $S$ is convex in $\mathcal K$.
\end{thm}

\subsection{Poset with inconsistent pairs (PIP)}\label{subsec:pip}
Barth\'{e}lemy and Constantin~\cite{Bar} established a Birkhoff-type representation theorem for median semilattices, i.e., pointed median graphs, by employing a poset with an additional relation.
This structure was rediscovered by Ardila et al.~\cite{Ard} in the context of CAT(0) cubical complexes.
An \textit{antichain} of a poset $P$ is a subset of $P$ that contains no two comparable elements.
A subset $I$ of $P$ is called an \textit{order ideal} of $P$ if $a \in I$ and $b \preceq a$ imply $b \in I$.
A poset $P$ is \textit{locally finite} if every interval $[a, b] = \set{c \in P}{a \preceq c \preceq b}$ is finite, and it has \textit{finite width} if every antichain is finite.
\begin{defi}
A \textit{poset with inconsistent pairs} (or, briefly, a \textit{PIP}) is a locally finite poset $P$ of finite width, endowed with a symmetric binary relation $\smallsmile$ satisfying:
\begin{enumerate}
\renewcommand{\labelenumi}{\theenumi}
\renewcommand{\theenumi}{\arabic{enumi})}
\setlength{\parskip}{0cm}
\setlength{\itemsep}{0cm}
	\item If $a \smallsmile b$, then $a$ and $b$ are incomparable.
	\item If $a \smallsmile b$, $a\preceq a'$ and $b \preceq b'$, then $a' \smallsmile b'$.
\end{enumerate}
A pair $\{ a, b \}$ with $a \smallsmile b$ is called an \textit{inconsistent pair}. An order ideal of $P$ is called \textit{consistent} if it contains no inconsistent pairs.
\end{defi}
For a CAT(0) cubical complex $\mathcal K$ and a vertex $v$ of $\mathcal K$, the pair $(\mathcal K, v)$ is called a \textit{rooted CAT(0) cubical complex}.
Given a poset with inconsistent pairs $P$, one can construct a cubical complex $\mathcal K_P$ as follows:
The underlying graph $G(\mathcal K_P)$ is a graph $G_P$ whose vertices are consistent order ideals of $P$ and where two consistent order ideals $I, J$ are
adjacent if and only if $|I \Delta J| = 1$; replace all cube-subgraphs (i.e., subgraphs isomorphic to cubes of some dimensions) of $G_P$ by solid cubes.
See Figure~\ref{fig:pip} for an example.
\begin{figure}[t]
	\centering
	\includegraphics[bb= 0 0 363 128]{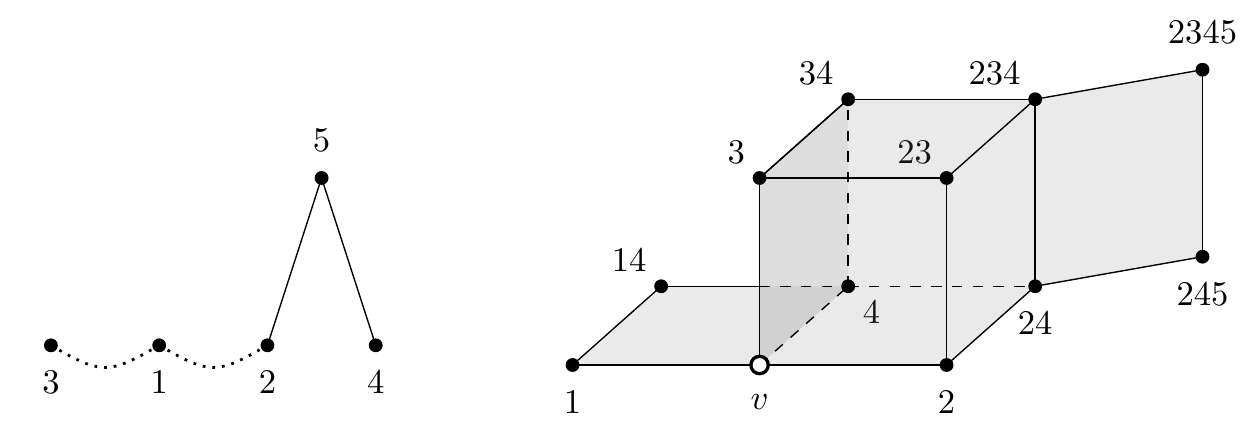}
	\caption{A poset with inconsistent pairs and the corresponding rooted CAT(0) cubical complex. Dotted line represents minimal inconsistent pairs, where an inconsistent pair $\{a, b \}$ is said to be \textit{minimal} if there is no other inconsistent pair $\{a', b'\}$ with $a' \preceq a$ and $b' \preceq b$.}
	\label{fig:pip}
\end{figure}
In fact, the resulting cubical complex $\mathcal K_P$ is CAT(0), and moreover:
\begin{thm}[{Ardila et al.~\cite{Ard}}]
The map $P \mapsto \mathcal K_P$ is a bijection between posets with inconsistent pairs and rooted CAT(0) cubical complexes.
\end{thm}
This bijection can also be derived from Theorem~\ref{thm:Chepoi} and the result of Barth\'{e}lemy and Constantin~\cite{Bar}, who found a bijection between PIPs and pointed median graphs.

Given a poset with inconsistent pairs $P$, one can embed $\mathcal K_P$ into a unit cube in the Euclidean space as follows, which we call the \textit{standard embedding} of $P$~\cite{Ard}:
\begin{equation*}
	\mathcal K_P = \set{(x_i)_{i\in P} \in [0, 1]^{P}}{i \prec j \ \text{and}\  x_i < 1 \Ra x_j = 0, \, \text{and}\  i \smallsmile j \Ra x_ix_j = 0}.
\end{equation*}
For each pair $(I, M)$ of a consistent order ideal $I$ of $P$ and a subset $M \subseteq I_{\max}$,
where $I_{\max}$ is the set of maximal elements of $I$,
the subspace 
\begin{equation*}
C_M^I := \set{x \in \mathcal K_P}{ i \in I \bsl M  \Ra x_i = 1, \, \text{and} \ i \notin I \Ra x_i = 0} = \{1\}^{I\bsl M} \times [0, 1]^M \times \{0\}^{P \bsl I}
\end{equation*}
corresponds to a unique $|M|$-dimensional cell of $\mathcal K_P$.

\subsection{CAT(0) orthant space}\label{subsec:orthant}
Let $\real_+$ denote the set of nonnegative real numbers.
Let $\mathcal L$ be an abstract simplicial complex on a finite set $V$.
The \textit{orthant space} $\mathcal O(\mathcal L)$ for $\mathcal L$ is a subspace of $|V|$-dimensional orthant $\mathcal \real_+^V$ constructed by taking a union
of all subcones $\set{\mathcal O_S}{S \in \mathcal L}$ associated with simplices of $\mathcal L$,
where $\mathcal O_S$ is defined by $\mathcal O_S := \real_+^S \times \{ 0 \}^{V \bsl S}$ for each simplex $S \in \mathcal L$;
namely, $\mathcal O(\mathcal L) = \bigcup_{S \in \mathcal L} \set{x \in \real_+^V}{x_v = 0 \ \text{for each} \ v \notin S}$.
The distance between two points $x, y \in \mathcal O(\mathcal L)$ is defined in a similar way as in the case of cubical complexes.
An orthant space is a special instance of cubical complexes.
\begin{thm}[{Gromov~\cite{Gro}}]
The orthant space $\mathcal O(\mathcal L)$ for an abstract simplicial complex $\mathcal L$ is a CAT(0) space if and only if $\mathcal L$ is a flag complex.
\end{thm}

A typical example of CAT(0) orthant spaces is a \textit{tree space}~\cite{Bil}.
Owen and Provan~\cite{Owe, OwePro} gave a polynomial time algorithm to compute geodesics in tree spaces,
which was generalized to CAT(0) orthant spaces by Miller et al.~\cite{Mil}.
\begin{thm}[{\cite{Mil, Owe, OwePro}}]\label{thm:Owen}
Let $\mathcal L$ be a flag abstract simplicial complex on a finite set $V$ and $\mathcal O (\mathcal L)$ be the CAT(0) orthant space for $\mathcal L$.
Let $x, y \in \mathcal O (\mathcal L)$, and let $S_1$ and $S_2$ be the inclusion-wise minimal simplices such that $x \in \mathcal O_{S_1}$ and $y \in \mathcal O_{S_2}$.
Then one can find the explicit description of the geodesic joining $x$ and $y$ in $O((|S_1|+|S_2|)^4)$ time.
\end{thm}
An interesting thing about their algorithm is that it solves as a subproblem a combinatorial optimization problem: a Maximum Weight Stable Set problem on a bipartite graph whose color
classes have at most $|S_1|, |S_2|$ vertices, respectively.
We should note that the above explicit descriptions of geodesics are radical expressions.
Computationally, for a point $p$ on a geodesic,
one can compute a rational point $p' \in \mathcal O (\mathcal L)$ such that  $d(p', p) \le \delta$ and the number of bits required for each coordinate of $p'$ is bounded by $O(\log(|V|/\delta))$.

For a CAT(0) orthant space $\mathcal O(\mathcal L)$ and a real number $r > 0$,
we call $\mathcal O(\mathcal L)|_{[0, r]} := \mathcal O(\mathcal L) \cap [0, r]^V$ a \textit{truncated CAT(0) orthant space}.
As a consequence of Theorem~\ref{thm:Owen}, one obtains the following:
\begin{thm}[{\cite{Mil}}]\label{thm:Miller}
Let $\mathcal L$ be a flag abstract simplicial complex on a finite set $V$ and $\mathcal O (\mathcal L)|_{[0,r]}$ be a truncated CAT(0) orthant space for $\mathcal L$.
Then for any two points $x, y \in \mathcal O (\mathcal L)|_{[0,r]}$, one can find the explicit description of the geodesic joining $x$ and $y$ in $O(|V|^4)$ time.
\end{thm}
In fact, a truncated CAT(0) orthant space $\mathcal O(\mathcal L)|_{[0,r]}$ is a convex subspace of $\mathcal O(\mathcal L)$.

\subsection{Main theorem}\label{subsec:main}
We now consider the following problem.
It should be remarked that as stated in~\cite{Ard} there are no simple formulas for the breakpoints in geodesics in CAT(0) cubical complexes due to their algebraic complexity,
and hence one can only compute them approximately.
Computationally, we adopt the standard embedding as an input CAT(0) cubical complex.
\begin{prob}\label{prob:1}
Given a poset with inconsistent pairs $P$, two points $p, q$ in the standard embedding $\mathcal K_P$ of $P$, and a positive parameter $\epsilon > 0$,
find a sequence of points $p = x_0, x_1, \dots, x_{n-1}, x_n = q$ in $\mathcal K_P$ with $\sum_{i =0}^{n-1} d(x_{i}, x_{i+1}) \le d(p, q) + \epsilon$ and compute the
geodesic joining $x_{i}$ and $x_{i+1}$ for $i =0, 1, \dots, n-1$.
\end{prob}
Our main result is the following theorem.
Note that for the shortest path problem in a general CAT(0) cubical complex there has been no algorithm that runs in time polynomial in the size of the complex, much less the size of the compact representation PIP.

\begin{thm}\label{thm:main2}
Problem~\ref{prob:1} can be solved in $O(|P|^7\log(|P|/\epsilon))$ time.
Moreover, the number of bits required for each coordinate of points in $\mathcal K_P$ occurring throughout the algorithm can be bounded by $O(\log(|P|/\epsilon))$.
\end{thm}

Let us show this theorem. Let $m$ denote the number of elements of $P$ and let $D < 1$ be a positive constant close to $1$ (e.g., set $D := 0.9$).
Theorem~\ref{thm:main1} implies that in order to prove Theorem~\ref{thm:main2} it suffices to show that:
\begin{enumerate}
\renewcommand{\labelenumi}{\theenumi}
\renewcommand{\theenumi}{(\alph{enumi})}
\setlength{\parskip}{0cm}
\setlength{\itemsep}{0cm}
	\item\label{item:a} Given two points $p, q \in \mathcal K_P$, one can find a sequence of points $p = x_0, x_1, \dots, x_{n-1}, x_n = q$ in $\mathcal K_P$ such that $n = O(m)$ and $d(x_i, x_{i+1}) \le D/4-\epsilon$ for $i = 0, 1, \dots, n-1$.
	\item\label{item:b} Given two points $p, q \in \mathcal K_P$ with $d(p, q) \le D$, one can compute the geodesic joining $p$ and $q$ in $O(m^4)$ time and find a $\delta$-midpoint $w$ of $p$ and $q$ with $O(\log(m/\delta))$ bits enough for each coordinate of $w$.
\end{enumerate}

It is relatively easy to show \ref{item:a}, by considering a curve $c(p, q)$ issuing at $p$, going through an edge geodesic (a shortest path in the underlying graph of $\mathcal K_P$) between some vertices of cells containing $p, q$, and ending at $q$.
(Note that one can easily find an edge geodesic between vertices $u$ and $v$ of $\mathcal K_P$ under the PIP representation.
Reroot the complex $\mathcal K_P$ at $u$. In other words, construct a poset $P'$ for which $\mathcal K_{P'} \cong \mathcal K_P$ and $u$ is the root of $\mathcal K_{P'}$; this construction is implicitly stated in~\cite{Ard}.
Then the edge geodesic in $\mathcal K_{P'}$ from the root $u = \emptyset$ to $v = I$, where $I$ is a consistent order ideal of $P'$, can be found by considering a linear extension of the elements of $I$.)
Since such a curve $c(p, q)$ has length at most $O(m)$, dividing it into parts appropriately, one can get a desired sequence of points.
To show \ref{item:b}, we need the following two lemmas.
\begin{lem}\label{lem:star}
Let $\mathcal K$ be a CAT(0) cubical complex and $v$ be a vertex of $\mathcal K$.
Then the star $\mathrm{St}(v, \mathcal K)$ of $v$ in $\mathcal K$, i.e., the subcomplex spanned by all cells containing $v$, is convex in $\mathcal K$.
\end{lem}
\begin{proof}
Let $G^{\Delta}$ be the graph having the same vertex set as $G = G(\mathcal K)$, where two vertices are adjacent if and only if they belong to a common cube of $G$.
It is well-known that every ball $B(u, r) := \{ u' \in V(G^{\Delta}) \, | \, d_{G^{\Delta}}(u, u') \le r \}$ of $G^{\Delta}$ is a convex set in a median graph $G$; see, e.g.,~\cite[Proposition 2.6]{Ban}.
In particular, the ball $B(v, 1)$ of $G^{\Delta}$, which coincides with the vertex set of $\mathrm{St}(v, \mathcal K)$, is convex in $G$.
Hence, by Theorem~\ref{thm:CheMaf}, $\mathrm{St}(v, \mathcal K)$ is convex in $\mathcal K$ in the $\ell_2$-metric.
\end{proof}

\begin{lem}\label{lem:disjoint}
Let $\mathcal K$ be a CAT(0) cubical complex.
Let $p, q$ be two points in $\mathcal K$ with $d(p, q) < 1$ and $R_1, R_2$ be the minimal cells of $\mathcal K$ containing $p, q$, respectively.
Then $R_1 \cap R_2 \neq \emptyset$.
\end{lem}

We give a proof of Lemma~\ref{lem:disjoint} in Section~\ref{subsec:proof}.
Using these lemmas, we show \ref{item:b}.
Suppose that we are given two points $p, q \in \mathcal K_P$ with $d(p, q) \le D$.
First notice that one can find in linear time the minimal cells $R_1$ and $R_2$ of $\mathcal K_P$ that contain $p$ and $q$, respectively, just by checking their coordinates.
(Indeed, one has $R_1 = C_M^I$ for $I = \set{i \in P}{p_i > 0}$ and $M = \set{i \in P}{0 < p_i < 1}$.)
Since $d(p, q) \le D < 1$, from Lemma~\ref{lem:disjoint} we know that $R_1 \cap R_2 \neq \emptyset$.
Let $v$ be a vertex of $R_1 \cap R_2$.
Then $p$ and $q$ are contained in the star $\mathrm{St}(v, \mathcal K_P)$ of $v$.
Since $\mathrm{St}(v, \mathcal K_P)$ is convex in $\mathcal K_P$ by Lemma~\ref{lem:star}, we only have to compute the geodesic in $\mathrm{St}(v, \mathcal K_P)$.
Obviously, $\mathrm{St}(v, \mathcal K_P)$ is a truncated CAT(0) orthant space, and hence
one can compute the geodesic between $p$ and $q$ in $\mathrm{St}(v, \mathcal K_P)$ in $O(m^4)$ time, by Theorem~\ref{thm:Miller}.
In addition, one can find a $\delta$-midpoint $w \in \mathrm{St}(v, \mathcal K_P)$ of $p$ and $q$ such that the number of bits required for each coordinate of $w$ is bounded by $O(\log(m/\delta))$.
This implies \ref{item:b} and therefore completes the proof of Theorem~\ref{thm:main2}.


\subsection{Proof of Lemma~\ref{lem:disjoint}}\label{subsec:proof}
We end this paper by giving a proof of Lemma~\ref{lem:disjoint}, which were used in proving Theorem~\ref{thm:main2}.
We start with some properties of halfspaces in a median graph $G = (V, E)$.
For any edge $ab$ of $G$, define $H(a, b) := \set{v \in V}{d_G(a, v) <d_G(b, v)}$ and $H(b, a) := \set{v \in V}{d_G(b, v) <d_G(a, v)}$.
Then $H(a, b)$ and $H(b, a)$ are complementary halfspaces of $G$~\cite{Mul}.
The \textit{boundary} of a halfspace $H$ of $G$, denoted by $\partial H$, consists of all vertices of $H$ which have a neighbor in the complement $H' := V\bsl H$ of $H$.
(Note that such a neighbor is unique for each vertex in $\partial H$, since median graphs are bipartite.)
Thus, one can define a bijection $\phi_H : \partial H \to \partial H'$ such that $v = \phi_H(u)$ if and only if there exists an edge $uv$ with $u \in \partial H$ and $v \in \partial H'$.
It was shown by Mulder~\cite{Mul} that the boundaries $\partial H$ and $\partial H'$ of complementary halfspaces $H, H'$
induce convex subgraphs of $G$ with an isomorphism $\phi_H$. Hence $\partial H \cup \partial H'$ and $\partial H \cup  H'$ are also convex sets of $G$.

We recall some basic properties of CAT(0) spaces.
Let $X$ be a CAT(0) space, and let $Y$ be a complete closed convex subset of $X$.
Then for every $x \in X$, there exists a unique point $\pi(x) \in Y$ such that $d(x, \pi(x)) = d(x, Y) := \inf_{y \in Y}d(x, y)$.
The resulting map $\pi : X \to Y$ is called the \textit{orthogonal projection} onto $Y$; see~\cite{Bri} for details.

\begin{prop}\label{prop:l1=l2}
Let $\mathcal K$ be a CAT(0) cubical complex.
Let $C$ and $v$ be a cell and a vertex of $\mathcal K$, respectively.
Then the gate of $v$ in $V(C)$ in the graph $G(\mathcal K)$ coincides with the image $\pi(v)$ of $v$ under the orthogonal projection $\pi : \mathcal K \to C$ onto $C$.
\end{prop}
\begin{proof}
Let $v'$ be the gate of $v$ in $V(C)$ in the graph $G = G(\mathcal K)$.
Let $v_1, v_2, \dots, v_k$ be the neighbors of $v'$ in $V(C)$, where $k$ is the dimension of $C$.
Let us write $H_i := H(v', v_i) = \{ u \in V(G) \, | \, d_G(v', u) < d_G(v_i, u) \}$ for each $i = 1, 2, \dots, k$.
Note that each halfspace $H_i$ contains $v$, as $v' \in I_G(v_i, v)$.

We show $\pi(v) \in C \cap \mathcal K(H_i)$ for each $i = 1, 2, \dots, k$.
To see this, fix an arbitrary $i$ and set $H := H_i$ and $H' := V(G) \bsl H_i$.
Let $\gamma$ be the geodesic in $\mathcal K$ with $\gamma(0) = v$ and $\gamma(1) = \pi(v)$.
Then one can find a point $p := \gamma(s)$ on the geodesic for some $s \in [0, 1]$ that belongs to $\mathcal K(\partial H)$.
As remarked above, $\partial H \cup \partial H'$ is convex in $G$, and hence $\mathcal K(\partial H \cup \partial H')$ is convex in $\mathcal K$ by Theorem~\ref{thm:CheMaf}.
This implies that the geodesic segment joining $p$ and $\pi(v)$ is contained in $\mathcal K(\partial H \cup \partial H')$.
Note that $\mathcal K(\partial H \cup \partial H')$ is isometric to $\mathcal K(\partial H) \times [0, 1]$.
Let $\psi : \mathcal K(\partial H \cup \partial H') \to \mathcal K(\partial H) \times [0, 1]$ be the isometry that sends $\mathcal K(\partial H)$ to $\mathcal K(\partial H) \times \{0 \}$
and $\mathcal K(\partial H')$ to $\mathcal K(\partial H) \times \{1\}$.
For each point $y$ in $K(\partial H \cup \partial H')$, when writing its image $\psi(y)$ as $(y_1, y_2) \in \mathcal K(\partial H) \times [0, 1]$,
we shall write $y_H$ to denote the point $\psi^{-1}((y_1, 0))$ in $\mathcal K(\partial H)$.
Let $\gamma' : [0, 1] \to \mathcal K$ be the map obtained from $\gamma$ by reseting $\gamma'(t) := (\gamma(t))_H$ for all $t \in [s, 1]$; see Figure~\ref{fig:proj} for intuition.
\begin{figure}[t]
	\centering
	\includegraphics[bb= 0 0 366 167]{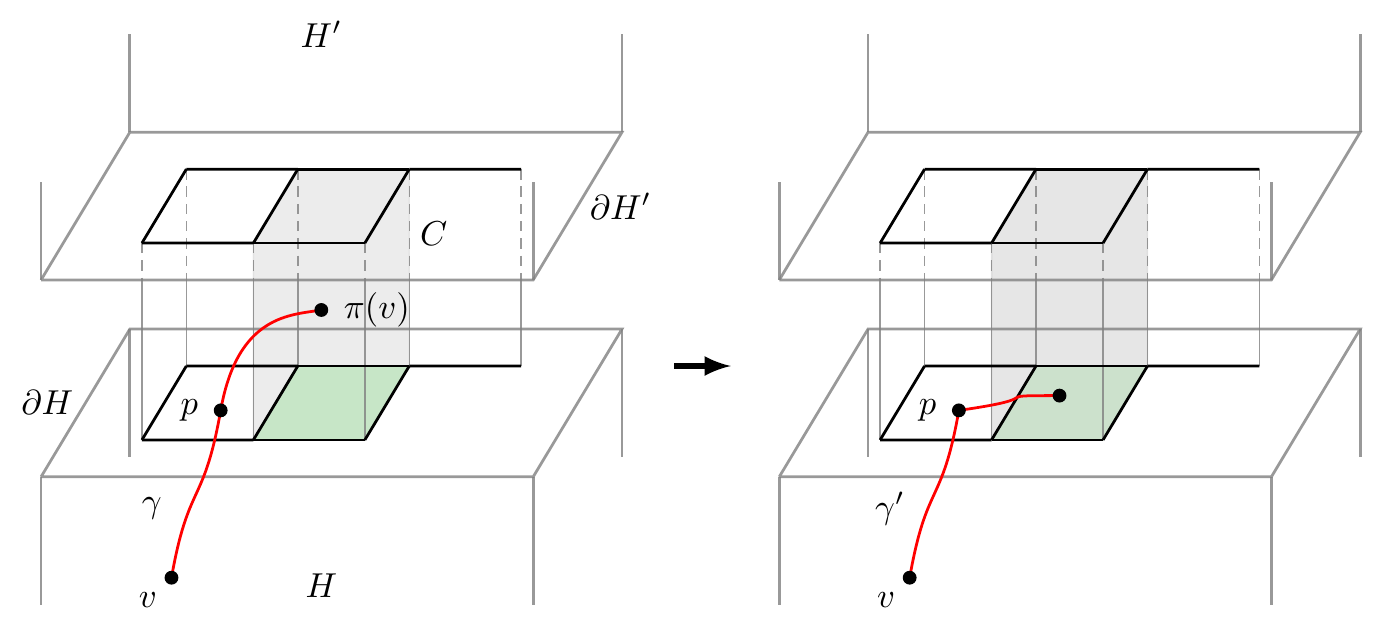}
	\caption{To the proof of Proposition \ref{prop:l1=l2}. Taking some $s \in [0, 1]$ such that $p := \gamma(s)$ lies in $\mathcal K(\partial H)$, and ``projecting'' the point $\gamma(t)$ onto
	$\mathcal K(\partial H)$ for all $t \in [s, 1]$, one can get a path $\gamma'$ joining $v$ and some point on $C \cap \mathcal K(H)$ whose length is at most that of $\gamma$.}
	\label{fig:proj}
	\end{figure}
Then $\gamma'$ is a continuous map in $\mathcal K$ joining $\gamma'(0) = v$ and $\gamma'(1) \in C \cap \mathcal K(H)$ whose length is at most the length of $\gamma$.
This implies that $\pi(v)$ should belong to $C \cap \mathcal K(H)$. Thus, we obtain $\pi(v) \in C \cap \mathcal K(H_i)$ for each $i = 1, 2, \dots, k$.

Note that the intersection of all $C \cap \mathcal K(H_i)$ is a singleton $\{v' \}$.
This implies that $\pi(v) \in \{v' \}$ and completes the proof.
\end{proof}

Now let us show Lemma~\ref{lem:disjoint}.
Let $p, q$ be two points in $\mathcal K$ with $d(p, q) < 1$ and $R_1, R_2$ be the minimal cells of $\mathcal K$ containing $p, q$, respectively.
Let $u \in V(R_1)$ and $v \in V(R_2)$ be vertices satisfying $d_G(u, v) = \min \{ d_G(u', v') \, | \, u' \in V(R_1), v' \in V(R_2) \}$ in the underlying graph $G = G(\mathcal K)$.
It is easy to see that $u$ is the gate of $v$ in $V(R_1)$ and $v$ is the gate of $u$ in $V(R_2)$, in the graph $G$.
Hence by Proposition~\ref{prop:l1=l2} we have $\pi_1(v) = u$ and $\pi_2(u) = v$, where $\pi_i : \mathcal K \to R_i$ is the orthogonal projection onto $R_i$ for $i = 1, 2$.
This implies that $d(u, v) = d(R_1, R_2) := \inf \{d(x, y) \, | \, x \in R_1, y \in R_2 \}$.
Since $d(R_1, R_2) \le d(p, q) < 1$, we have $d(u, v) < 1$.
Hence $u$ and $v$ should be the same vertex of $\mathcal K$, and thus, we have $R_1 \cap R_2 \neq \emptyset$.
This completes the proof of Lemma~\ref{lem:disjoint}.


\section*{Acknowledgments}
I thank Hiroshi Hirai for introducing me to this problem and for helpful comments and careful reading.
The work was supported by JSPS KAKENHI Grant Number 17K00029, and by JST ERATO Grant Number JPMJER1201, Japan.

\end{document}